%% file: 00-main.tex
\documentclass[a4paper,UKenglish,cleveref,thm-restate]{lipics-v2021}

\usepackage[subrefformat=simple,labelformat=simple,position=b]{subcaption}

\usepackage{nicefrac}
\usepackage{mathtools}
\usepackage{xspace}
\usepackage{complexity}
\usepackage{cleveref}

\hideLIPIcs

\graphicspath{{./figures/}}

\title{Coordinated Motion Planning: Multi-Agent Path Finding in a Densely Packed, Bounded Domain}
\titlerunning{Multi-Agent Path Finding in a Densely Packed, Bounded Domain}

\author{Sándor P. Fekete}{Department of Computer Science, TU Braunschweig, Braunschweig, Germany}{s.fekete@tu-bs.de}{https://orcid.org/0000-0002-9062-4241}{}
\author{Ramin Kosfeld}{Department of Computer Science, TU Braunschweig, Braunschweig, Germany}{r.kosfeld@tu-bs.de}{https://orcid.org/0000-0002-1081-2454}{}
\author{Peter Kramer}{Department of Computer Science, TU Braunschweig, Braunschweig, Germany}{p.kramer@tu-bs.de}{https://orcid.org/0000-0001-9635-5890}{}
\author{Jonas Neutzner}{Department of Computer Science, TU Braunschweig, Braunschweig, Germany}{j.neutzner@tu-bs.de}{}{}
\author{Christian Rieck}{Department of Computer Science, TU Braunschweig, Braunschweig, Germany}{rieck@ibr.cs.tu-bs.de}{https://orcid.org/0000-0003-0846-5163}{}
\author{Christian Scheffer}{Department of Electrical Engineering and Computer Science, Bochum University of Applied Sciences, Bochum, Germany}{christian.scheffer@hs-bochum.de}{https://orcid.org/0000-0002-3471-2706}{}

\authorrunning{S. P. Fekete, R. Kosfeld, P. Kramer, J. Neutzner, C. Rieck, and C. Scheffer}

\Copyright{Sándor P. Fekete, Ramin Kosfeld, Peter Kramer, Jonas Neutzner, Christian Rieck, Christian Scheffer}



\nolinenumbers

\begin{document}

    \maketitle

    \input{shortcuts}    \begin{abstract}
        We study \textsc{Multi-Agent Path Finding} for arrangements of labeled agents in the interior of a simply connected domain:
        Given a unique start and target position for each agent, the goal is to find a sequence of parallel, collision-free agent motions that minimizes the overall time (the~\emph{makespan}) until all agents have reached their respective targets.
        A natural case is that of a simply connected~polygonal domain with axis-parallel boundaries and integer coordinates, i.e., a~\emph{simple polyomino}, which amounts to a simply connected union of lattice unit squares or \emph{cells}.
        We focus on the particularly challenging setting of densely packed agents, i.e., one per cell, which strongly restricts the mobility of agents, and requires intricate coordination of motion.

        We provide a variety of novel results for this problem, including
        (1) a characterization of polyominoes in which a reconfiguration plan is guaranteed to exist;
        (2) a characterization of shape parameters that induce worst-case bounds on the makespan;
        (3) a suite of algorithms to achieve asymptotically worst-case optimal performance with respect to the achievable \emph{stretch} for cases with severely limited maneuverability.
        This corresponds to bounding the ratio between obtained makespan and the lower bound provided by the max-min distance between the start and target position of any agent and our shape parameters.

        Our results extend findings by Demaine et al.~\cite{dfk+-cmprs-18,dfk+-cmprs-19} who investigated the problem for solid rectangular domains, and in the closely related field of \textsc{Permutation Routing}, as presented by Alpert et al.~\cite{ALPERT2022101862} for convex pieces of grid graphs.
    \end{abstract}
    \newpage

    \input{01-introduction}
    \input{02-solvability}

    \input{03-obstacles}
    \input{04-upper-bounds}
    \input{05-conclusion}

    \bibliography{references}
\end{document}

%% file: shortcuts.tex
\DeclarePairedDelimiter\ceil{\lceil}{\rceil}
\DeclarePairedDelimiter\floor{\lfloor}{\rfloor}

\newcommand{\configurations}{\mathcal{C}}
\newcommand{\BigO}{\mathcal{O}}
\newcommand{\smallo}{\text{\scalebox{.7}{$\mathcal{O}$}}}

\newcommand{\diam}{\ifmmode d\else diameter\xspace\fi}
\newcommand{\applicable}{applicable\xspace}
\newcommand{\bottleneck}{\ifmmode\zeta\else bottleneck\xspace\fi}
\newcommand{\Bottleneck}{\ifmmode\zeta\else Bottleneck\xspace\fi}
\newcommand{\bottleneckfrac}[2][P]{\ifmmode\nicefrac{\bottleneck(#1)}{#2}\else$\nicefrac{\bottleneck(#1)}{#2}$\fi}
\newcommand{\magicConstant}{\tau}

\newcommand{\depth}{\ifmmode\mu\else depth\xspace\fi}
\newcommand{\Depth}{\ifmmode\mu\else Depth\xspace\fi}
\newcommand{\scale}{\ifmmode{c}\else scale\xspace\fi}
\newcommand{\dual}{G}

\newcommand{\Hinterland}{\ifmmode\bottleneckfrac{4}\else Watershed\xspace\fi}
\newcommand{\hinterland}{\ifmmode\bottleneckfrac{4}\else watershed\xspace\fi}
\newcommand{\skeleton}{\ifmmode\lambda\else skeleton\xspace\fi}
\newcommand{\Skeleton}{\ifmmode\lambda\else Skeleton\xspace\fi}
\newcommand{\bottleneckHalfCover}{\square_{2\skeleton}}
\newcommand{\tile}{tile\xspace}
\newcommand{\tiles}{tiles\xspace}

\newcommand{\rotatesort}{\textsc{RotateSort}\xspace}
\newcommand{\bubblesort}{\textsc{Parallel BubbleSort}\xspace}
\newcommand{\solvable}{universally reconfigurable\xspace}
\newcommand{\solvability}{universal reconfigurability\xspace}

%% file: 01-introduction.tex
\section{Introduction}
\label{sec:introduction}
Problems of coordinating the motion of a set of objects occur in a wide range of applications, such as warehouses~\cite{wurman2008coordinating}, multi-agent motion planning~\cite{rubenstein2014programmable,sw-sr-08}, and aerial swarm robotics~\cite{cpdsk-saesr-18}.
In \textsc{Multi-Agent Path Finding} (MAPF)~\cite{stern2019multi}, we are given a set of agents, each with an initial and a desired target position within a certain domain.
The task  is to determine a coordinated motion plan: a sequence of parallel, collision-free movements such that the time by which all agents have reached their destinations (the \emph{makespan}) is minimized.

Theoretical aspects of MAPF have enjoyed significant attention.
In the early days of computational geometry, Schwartz and Sharir~\cite{ss-pmpcbpb-83} developed methods for coordinating the motion of disk-shaped objects between obstacles, with runtime polynomial in the complexity of the obstacles, but exponential in the number of disks.
The fundamental difficulty of geometric MAPF was highlighted by Hopcroft et al.~\cite{hss-cmpmio-84,hw-rmompgs-86}, who showed that it is \PSPACE\nobreakdash-complete to decide whether multiple agents can reach a given target configuration.
In contrast, closely related graph-based variants of the MAPF problem permit for linear time algorithms for the same decision problem~\cite{yu2015pebble}.

More recently, Demaine et al.~\cite{bfk+-cmpv-18,dfk+-cmprs-18,dfk+-cmprs-19} have provided methods to compute \emph{constant stretch} solutions for coordinated motion planning in unbounded environments in which agents occupy distinct grid cells.
The stretch of a solution is defined as the ratio between its makespan and a trivial lower bound, the \emph{\diam} $\diam$, which refers to maximum distance between any agent's origin and destination.
Their work therefore obtains collision-free motion schedules that move each agent to its target position in $\BigO(\diam)$ discrete moves, which corresponds to a constant-factor approximation.
However, their methods assume the absence of a environmental boundary that may impose external constraints on the agents' movements.%

\subsection{Our contributions}
\label{subsec:our-contributions}
In this paper, we study \textsc{Multi-Agent Path Finding} for densely packed arrangements of labeled agents that are required to remain within discrete grid domains, i.e., polyominoes.
This is a natural constraint that occurs in many important applications, but provides considerable additional difficulties; in particular, a coordinated motion
plan may no longer exist for domains with narrow bottlenecks.
We provide a variety of novel contributions:
\medskip
\begin{itemize}
	\item {
		We give a full characterization of simple polyominoes $P$ that are \solvable.
		These allow \emph{some} feasible coordinated motion plan for \emph{any} combination of initial and desired target configurations, without regard for the makespan: We prove that this is the case if and only if $P$ has a cover by $2\times 2$ squares with a connected intersection graph.
	}
	\item {
		We model the shape parameters \emph{\bottleneck length}~$\bottleneck(P)$ (which is the minimum length of a cut dividing the region into non-trivial pieces) and \emph{domain \depth}~$\depth(P)$ (which is the maximum distance of any cell from the domain boundary).
		We provide refined upper and lower bounds on the makespan and stretch factor based on these shape parameters.
	}
	\item {
		For some instances, any \applicable schedule may require
		a makespan of~$\Omega(\diam + \nicefrac{\diam^2}{\bottleneck(P)})$.
		We show how to compute schedules of makespan linear in the ratio of domain area and \bottleneck, i.e., $\BigO(\nicefrac{n}{\bottleneck(P)})$.
	}
	\item {
		We characterize \emph{narrow} instances, which feature very limited \depth relative to the diameter~$\diam$, and provide an approach for asymptotically worst-case optimal schedules.
	}
\end{itemize}
\pagebreak

\subsection{Related work}
\label{subsec:related-work}
\subparagraph*{Motion planning.}
\textsc{Multi-Agent Path Finding} is a widely studied problem.
Due to space constraints, we restrict our description to the most closely related work.
For more detailed references, refer to the extensive bibliography in~\cite{dfk+-cmprs-19} and the mentioned surveys~\cite{cpdsk-saesr-18,sw-sr-08,stern2019multi}.

Of fundamental importance to our work are the results by Demaine et al.~\cite{dfk+-cmprs-19}, who achieved reconfiguration with constant stretch for the special case of rectangular domains.
A~key idea is to consider a partition of the rectangle into tiles whose size is linear in \diam~$\diam$.
These tiles can then be reconfigured in parallel. 
First, flow techniques are applied to shift agents into their target tile; afterward, agents are moved to their respective target positions.
Furthermore, they showed that computing the optimal solution is strongly \NP-complete.

Fekete et al.~\cite{unlabeled-connected-journal,labeled-connected-journal}
considered the unconstrained problem on the infinite grid with the additional
condition that the whole arrangement needs to be connected after every
parallel motion.
They considered both the labeled and the unlabeled version of the problem, providing polynomial-time algorithms for computing schedules with constant stretch for configurations of sufficient scale.
They also showed that deciding whether there is a reconfiguration schedule with a makespan of $2$ is already \NP-complete, unlike deciding the same for a makespan of~$1$.

Eiben, Ganian, and Kanj~\cite{EibenGK23} investigated the parameterized complexity of the problem for the variants of minimizing the makespan and minimizing the total travel distance.
They analyzed the problems with respect to two parameters: the number of agents, and the objective target.
Both variants are \FPT\ when parameterized by the number of agents, while minimizing the makespan becomes para-\NP-hard when parameterized by the objective target.

Further related work studies (unlabeled) multi-robot motion planning problems in polygons.
Solovey and Halperin~\cite{SoloveyH16} show that the unlabeled variant is \PSPACE-hard, even for the specific case of unit-square robots moving amidst polygonal obstacles.
Even in simple polygonal domains, a feasible motion-plan for unlabeled unit-disc robots does not always exist, if, e.g., the robots and their targets are positioned too densely.
However, if there is some minimal distance separating start and target positions, Adler et al.~\cite{AdlerBHS15}~show that the problem always has a solution that can be computed efficiently.
Banyassady et al.~\cite{BanyassadyBBBFH22} prove tight separation bounds for this case.
Agarwal~et~al.~\cite{AgarwalGHT23} consider the labeled variant with revolving areas, i.e., empty areas around start and target positions.
They prove that the problem is \APX-hard, even when restricting to weakly-monotone motion plans, i.e., motion plans in which all robots stay within their revolving areas while an active robot moves to its target.
However, they also provide a constant-factor approximation algorithm.

The computational complexity of moving two distinguishable square-shaped robots in a polygonal environment to minimize the sum of traveled distances is still open;
Agarwal~et~al.~\cite{AgarwalHSS24} gave the first polynomial-time ($1+\varepsilon$)-approximation algorithm.

The problem was the subject of the 2021 CG:SHOP Challenge; see~\cite{Challenge22,Shadoks22,Liu22,Yang22} for an overview and a variety of practical computational methods and results.

\subparagraph*{Token swapping and routing via matchings.}
The task in the \textsc{Token Swapping Problem} is to transform two vertex labelings of a graph into one another by exchanging \emph{tokens} between adjacent vertices by sequentially selecting individual edges.
This problem is \NP-complete even for trees~\cite{token-swapping-tree}, and \APX-hard~\cite{token-swapping} in general.
Several approximation algorithms exist for different variants and classes of graphs~\cite{HeathV03,token-swapping,YamanakaDIKKOSS15}.
The \textsc{Permutation Routing} variant allows for parallelization, by selecting disjoint edge sets to perform swaps in parallel~\cite{AlonCG94,BaumslagA91,kawahara2019timecomplexity}.
The \emph{routing number} of a graph describes the maximal number of necessary parallel swaps between any two labelings.
Recently, Alpert et al.~\cite{ALPERT2022101862} presented an upper bound on the routing number of convex pieces of grid graphs, which is very closely related to our setting.

\subsection{Preliminaries}
\label{subsec:preliminaries}
We define the considered motion of \emph{agents} in a restricted environment~(\emph{domain}) as follows.

\subparagraph*{Domain.}
Consider the infinite integer grid graph, in which each $4$-cycle bounds a face of unit area, a \emph{cell}.
Every planar edge cycle in this grid graph bounds a finite set of cells, which induces a domain that we call a \emph{(simple) polyomino}, see~\cref{fig:boundary-cycle}.
We exclusively consider simple polyominoes, i.e., those without holes.
For the sake of readability, we might not state this for each individual polyomino in later sections.
The \emph{area} of a polyomino~$P$ is equal to the number of contained cells $n$.
The bounding edge cycle and its incident cells are therefore called is the \emph{boundary} and \emph{boundary cells} of $P$, respectively.
The dual graph of~$P$, denoted by $\dual(P)=(V,E)$, has a vertex for every cell, two of which are \emph{adjacent} if they share an edge in $P$, as shown in~\cref{fig:dual-graph}.
\begin{figure}[htb]
	\hfil%
	\begin{subfigure}[t]{0.32\textwidth}%
		\centering
		\includegraphics[page=1]{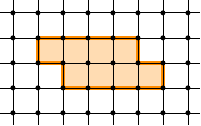}
		\caption{A boundary cycle.}
		\label{fig:boundary-cycle}
	\end{subfigure}
	\hfil%
	\begin{subfigure}[t]{0.32\textwidth}%
		\centering
		\includegraphics[page=2]{figures/preliminaries-domain}
		\caption{A polyomino's dual graph.}
		\label{fig:dual-graph}
	\end{subfigure}
	\hfil%
	\begin{subfigure}[t]{0.32\textwidth}%
		\centering
		\includegraphics[page=3]{figures/preliminaries-domain}
		\caption{A cut through a polyomino.}
		\label{fig:cut}
	\end{subfigure}
	\hfil%
	\caption{
		A polyomino, its dual graph, and a cut.
		Subsequent illustrations will only show the boundary and any relevant cuts, foregoing the underlying integer grid.
	}
	\label{fig:domain}
\end{figure}
The \emph{geodesic distance} between two cells of a polyomino~$P$ corresponds to the length of a shortest path between the corresponding vertices in $\dual(P)$.
A~\emph{cut} through a polyomino $P$ is defined by a planar path between two vertices on the boundary of $P$, as shown in~\cref{fig:cut}.
Its cut set corresponds exactly to those edges of~$\dual(P)$ which cross the path.
Geometrically, this induces two simple subpolyominoes $Q, R$ and write $Q,R\subset P$.
We say that a cut is \emph{trivial} if its endpoints on the boundary of $P$ have a connecting path on the boundary that is not longer than the cut itself.

\subparagraph*{Agents.}
We consider distinguishable \emph{agents} that occupy the cells of polyominoes.
A~\emph{configuration} of a polyomino $P$ with $\dual(P)=(V,E)$ is a bijective mapping $C:V\rightarrow \{1,\dots,n\}$ between cells and agent labels.
We denote the set of all configurations of $P$ as $\configurations(P)$.

In each discrete time step, an agent can either \emph{move}, changing its position $v$ to an adjacent position $w$, or hold its current position.
We denote this by $v\rightarrow w$ or $v\rightarrow v$, respectively.
Two parallel moves $v_1\rightarrow w_1$ and $v_2\rightarrow w_2$ are \emph{collision-free} if $v_1\neq v_2$ and $w_1\neq w_2$.
We assume that a \emph{swap}, i.e., two moves $v_1\rightarrow v_2$ and $v_2\rightarrow v_1$, causes a collision, and is therefore forbidden.
Configurations can be \emph{transformed} by sets of collision-free moves that are performed in parallel.
If a set of moves transforms a configuration $C_1$ into a configuration~$C_2$, this set is also called a \emph{transformation} $C_1\rightarrow C_2$.
For an illustrated example, see~\cref{fig:agents}.
A~\emph{schedule}~with \emph{makespan}~$M\in \mathbb{N}$ is then a sequence of transformations $C_1\rightarrow \dots \rightarrow C_{M+1}$, also denoted by $C_1\rightrightarrows C_{M+1}$.

\begin{figure}[htb]
	\hfil%
	\begin{subfigure}[t]{0.4\textwidth}%
		\centering
		\includegraphics[page=1]{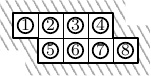}
		\caption{A polyomino and a configuration.}
		\label{fig:configuration}
	\end{subfigure}
	\hfil%
	\begin{subfigure}[t]{0.4\textwidth}%
		\centering
		\includegraphics[page=2]{figures/preliminaries-agents}
		\caption{Four agents move along a cycle.}
		\label{fig:transformation}
	\end{subfigure}
	\hfil%
	\caption{
		An illustration of an configurations and a transformations.
	}
	\label{fig:agents}
\end{figure}
\subparagraph*{Problem statement.}
We consider the \textsc{Multi-Agent Path Finding Problem} for agents in a discrete environment bounded by a simple polyomino.
Thus, an \emph{instance} of the problem is composed of two configurations $C_1,C_2\in\configurations(P)$ of a simple polyomino~$P$.
We say that a schedule is \emph{\applicable} to the instance exactly if it transforms $C_1$ into~$C_2$.
The~\emph{\diam of an instance} is the maximum geodesic distance~$\diam$ between an agents' start and target positions, and the \emph{stretch}~of an \applicable schedule is the ratio between its makespan and the \diam.%

%% file: 02-solvability.tex
\section{Reconfigurability}
\label{sec:reconfigurability}

In this section, we provide a characterization of (simple) polyominoes for which any configuration can be transformed into any other.
We say that a polyomino $P$ is \emph{\solvable} if there exists an \applicable schedule for any two configurations $C_1, C_2\in \configurations(P)$.
We prove that this is the case if and only if $P$ has a cover by cycles that have a connected intersection graph, and show how to compute an \applicable schedule of makespan~$\BigO(n)$.
\begin{theorem}
	\label{thm:parallel-bubblesort}
	A polyomino $P$ is \solvable~if and only if it has a cover by $2\times 2$ squares with a connected intersection graph.
	For any $C_1,C_2\in\configurations(P)$ of such a polyomino with area $n$, an \applicable schedule $C_1\rightrightarrows C_2$ of makespan $\BigO(n)$ can be computed efficiently.
\end{theorem}
Due to the cyclic nature of all movement, the edge connectivity of a polyomino's dual graph plays a significant role for \solvability.
We start with a negative~result.

\begin{lemma}
	\label{lem:simple-solvable-if-twofat}
	A polyomino $P$ that does not have a cover by $2\times 2$ squares with a connected intersection graph is not \solvable.
\end{lemma}
\begin{proof}
	All valid transformations have robots moving along cycles, taking the position of an adjacent robot that is leaving its cell in the same step.
	We thus propose that $P$ is not \solvable if its dual graph $\dual(P)=(V,E)$ is not $2$-edge-connected.

	Assume that $\dual(P)$ is not $2$-edge-connected.
	We can thus identify a cut-edge $e=\{v,w\}$.
	Clearly, there is no cycle in the dual graph of $P$ that involves the edge $e$, which implies that no robot can move from $v$ to $w$ or vice versa.
	Therefore, $P$ is not \solvable.

	Assume now that $P$ has no cover by $2\times 2$ squares, and consider some cell $v\in V$ which cannot be covered by a $2\times 2$ square in $P$.
	Let $\deg(v)$ be the number of cells adjacent to $v$.
	If $\deg(v)=1$, then $P$ is not $2$-edge-connected and therefore not \solvable.
	If $\deg(v)\geq 2$, we propose that at least two of the adjacent cells share a third adjacent cell:
	If no two adjacent cells to $v$ share a common neighbor $w$, it follows that at least one of the incident edges to $v$ is a cut-edge of the dual graph of $P$, implying that it is again not $2$-edge-connected, so $P$ is not \solvable.

	Finally, assume that $P$ has a cover by $2\times 2$ squares, but none with a connected intersection graph.
	This implies there exist two adjacent cells that cannot be simultaneously covered by the same $2\times 2$ square.
	Consider any two of such adjacent cells $v,w\in V$.
	If $v$ does not have an adjacent cell $x$ that is adjacent to a cell $y$, which itself is adjacent to $v$, then there is no alternate path from $v$ to $w$, unless there exists a hole in $P$.
	Since we consider only simple polyominoes, this means that $\{v,w\}$ is a cut-edge of the dual graph, and $P$ is not \solvable.
	This concludes our proof.
\end{proof}

Because direct swaps of adjacent agents are not possible, an important tool is the ability to ``simulate'' a large number of adjacent swaps in parallel, using a constant number of transformation steps.
Polyominoes that are unions of two $2\times 2$ squares form an important primitive to achieve this.
There exist two classes of such polyominoes; the squares can overlap either in one or two cells.
Clearly, either dual graph can be covered by two $4$-cycles that intersect in at least one vertex.
In~\cref{fig:two-overlapping-squares}, we illustrate schedules that perform adjacent swaps at the intersection of these cycles, implying \solvability of both classes.
In fact, any instance of either class takes at most $7$ or $14$ transformations, respectively.
\begin{figure}[htb]
	\begin{subfigure}[t]{\textwidth}
		\centering%
		\includegraphics{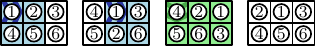}
		\caption{A schedule that swaps the robots labeled as $1$ and $2$.}
		\label{fig:two-overlapping-squares-horizontal}
	\end{subfigure}\vspace*{0.2cm}
	\begin{subfigure}[t]{\textwidth}
		\centering%
		\includegraphics{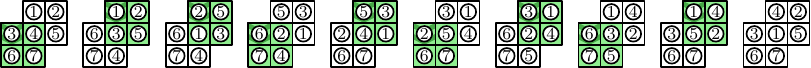}
		\caption{A schedule that swaps the robots labeled as $1$ and $4$.}
		\label{fig:two-overlapping-squares-diagonal}
	\end{subfigure}
	\caption{In polyominoes composed of two $2\times 2$ squares, we can realize swaps in $\BigO(1)$ steps.}
	\label{fig:two-overlapping-squares}
\end{figure}
\begin{observation}
	\label{obs:two-overlapping-squares}
	Polyominoes of two overlapping $2\times 2$ squares are \solvable.
\end{observation}

Using this observation for a primitive local operation, we show the following.

\begin{lemma}
	\label{lem:matching-realization}
	For any matching in the dual graph of a \solvable polyomino, we can compute a schedule of makespan $\BigO(1)$ which swaps the agents of all matched positions.
\end{lemma}
\begin{proof}
	Let $I$ refer to the connected intersection graph of a cover of a \solvable polyomino $P$ by $2\times 2$ squares, which can be computed in $\BigO(n)$.
	Due to \cref{lem:simple-solvable-if-twofat}, the vertices of each edge in the dual graph of $P$ share a common $2\times 2$ square in the cover.

	We thus divide the edges $E(I)$ of $I$ into $36$ classes based on each edge's orientation and $xy$-minimal coordinates mod $3$.
	These can be represented by $\{\uparrow, \nearrow, \rightarrow, \searrow\}\times [0,2]\times [0,2]$.

	For any intersection $\{u,v\}\in E(I)$, let $R(\{u,v\})$ now be the union of the vertices covered by the squares $u$ and $v$, which always corresponds exactly to one of the polyominoes outlined in \cref{obs:two-overlapping-squares}.
	Such a region has a bounding box no larger than $3\times 3$, which means that the regions $R(e)$ and $R(f)$ are disjoint for any two edges $e$ and $f$ in a common class, allowing us to apply \rotatesort in parallel to all the regions within one class.

	As there are constantly many classes, we can realize the adjacent swaps induced by a matching of adjacent cells in $\BigO(1)$ transformations.
\end{proof}

Marberg and Gafni~\cite{MarbergG88} propose an algorithm called \rotatesort that sorts a two-dimensional $n\times m$ array within $\BigO(n+m)$ parallel steps.
Demaine et al.~\cite{dfk+-cmprs-19} demonstrate that this algorithm can be applied geometrically, using the local swap mechanism illustrated in \cref{fig:two-overlapping-squares-horizontal}.
A geometric application of \rotatesort is a sequence of sets of pairwise disjoint adjacent swap operations, i.e., sets consisting of pairs of adjacent cells, where swaps can be simulated by circular rotations.
As our setting is not merely restricted to rectangular domains, we extend their approach using~\cref{lem:matching-realization}.
We give a constructive proof of \cref{thm:parallel-bubblesort} in the shape of an algorithm, as follows.

\begin{proof}[Proof of \cref{thm:parallel-bubblesort}]
	Our approach employs methods from \textsc{Permutation Routing}.
	In this setting, the task is to transform two different vertex labelings of a graph into one another by exchanging labels between adjacent vertices in parallel~\cite{AlonCG94}.
	A solution (or routing sequence) consists of a series of matchings, i.e., sets of independent edges, along which tokens are exchanged.
	Such a routing sequence of length at most $3n$ can be computed in almost linear time if the underlying graph is a tree~\cite{AlonCG94}.
	Thus, we consider an arbitrary spanning tree of a polyomino~$P$'s dual graph and compute such a routing sequence.
	Due to~\cref{lem:matching-realization}, each parallel swap operation in the sequence can be realized by a schedule of makespan $\BigO(1)$.
	We conclude that the schedule derived from the routing sequence has makespan $\BigO(n)$.
\end{proof}

%% file: 03-obstacles.tex
\section{The impact of the domain on the achievable makespan}
\label{sec:geometric-constraints}

Previous work has demonstrated that it is possible to achieve constant stretch for labeled agents in a rectangular domain.
However, in the presence of a non-convex boundary, such stretch factors may not be achievable.
We present~the following worst-case bound.
\begin{proposition}
    \label{prop:linear-stretch-bound}
    For any $\diam\geq 5$, there exist instances of \diam $\diam$ in \solvable polyominoes, such that all \applicable schedules have makespan $\Omega(\diam^2)$.
\end{proposition}

\begin{proof}
    We illustrate a class of such instances in~\cref{fig:linear-stretch}.
    In this class, agents located on different sides of a narrow passage must trade places.
    \cref{thm:parallel-bubblesort} tells us that these polyominoes are \solvable; however, any movement between the regions pass through the narrow passage at the center, limiting the number of agents exchanged between them to $2$ per transformation.
    As the number of agents scales quadratically with $\diam$, any schedule for this class of instances requires a makespan of~$\Omega(\diam^2)$.
\end{proof}
\begin{figure}[htb]
    \begin{subfigure}[b]{0.4\textwidth}%
        \centering%
        \includegraphics[scale=1.0,page=2]{bottleneck-and-depth}
        \caption{A polyomino with narrow passage.}
        \label{fig:narrow-passage-move}
    \end{subfigure}%
    \begin{subfigure}[b]{0.6\textwidth}%
        \centering%
        \includegraphics[scale=1.0,page=1]{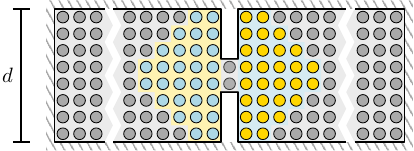}%
        \caption{This class can be extended such that $n > d^2$.}
        \label{fig:narrow-passage}
    \end{subfigure}%
    \caption{
        We illustrate a class of instances which require $\Omega(\diam^2)$ transformations.
    }
    \label{fig:linear-stretch}
\end{figure}

For a more refined characterization of features that affect the achievable makespan and to formulate a precise lower bound, we introduce the following shape parameter for polyominoes.

\medskip
\descriptionlabel{Bottleneck.}
We say that the \emph{\bottleneck} of a polyomino $P$ is the largest integer $\bottleneck(P)$, such that there is no non-trivial cut through $P$ of length less than $\bottleneck(P)$.
This means no interior ``shortcut'' of length less than $\bottleneck(P)$ exists between any two points on the boundary of~$P$.
A~\emph{bottleneck cut} through $P$ is therefore a non-trivial cut of length $\bottleneck(P)$.

\medskip
We now further refine the lower bound presented in~\cref{prop:linear-stretch-bound}, as follows.

\begin{proposition}
    \label{prop:bottleneck-stretch-bound}
    For any $\diam \geq 4$ and $z\in [2,\diam]$, there exists a \solvable polyomino~$P$ with $\bottleneck(P)=z$ that has instances of \diam $\diam$, for which any \applicable schedule has a makespan of $\Omega(\nicefrac{\diam^2}{\bottleneck(P)})$, i.e., a stretch factor of $\Omega(\nicefrac{\diam}{\bottleneck(P)})$.
\end{proposition}

\begin{proof}
    We formulate a generalized version of the instances from~\cref{prop:linear-stretch-bound}.
    By scaling the boundary of the polyomino between the two regions by an arbitrary amount less or equal to~$\diam$, we can create \solvable polyominoes with the targeted \bottleneck value.

    The movement between the two regions must then still be realized over the narrow grey region, limiting the number of robots exchanged between them to $\BigO(\bottleneck(P))$ per transformation.
    As the number of robots that need to traverse the \bottleneck cut scales quadratically with $\diam$, any \applicable schedule for this class of instances requires a makespan of~$\Omega(\nicefrac{\diam^2}{\bottleneck(P)})$.
\end{proof}

To further refine our understanding of the domain's impact on achievable makespans, we now characterize the size of widest passages, i.e., best case maneuverability.
To this end, we consider the maximum shortest distance to the boundary within the given domain, its \emph{\depth}.

\medskip
\descriptionlabel{Depth.}
We say that the \emph{\depth} of a polyomino $P$ is the smallest integer $\depth(P)$, such that every cell in $P$ has geodesic distance at most $\depth(P)$ to the boundary of $P$.
\medskip

The \depth and \bottleneck of a polyomino are very closely related, with \depth implying a bound on the \bottleneck of any (sub-)polyomino such that $\bottleneck(P')\leq2\depth(P)$ for any $P'\subseteq P$.
We take particular notice of the following property of \depth.
\begin{lemma}
    \label{lem:distance-to-non-trivial-cut}
    From any cell in a polyomino $P$, the maximal geodesic distance to a non-trivial geodesic cut of length at most $2\depth(P)$ is also at most $2\depth(P)$.
\end{lemma}
\begin{proof}
    For any subpolyomino $Q\subseteq P$ that is a square with side length $q$, the \depth of $P$ is at least~$\nicefrac{q}{2}$, as the center cell(s) of $Q$ are closer to the boundary of $P$ than $\nicefrac{q}{2}$.
    This, of course, implies that no such square can have a greater side length than $2\depth(P)$.
    We say that~$Q$ is a maximal square if there is no larger square $Q'$ such that $Q\subseteq Q'\subseteq P$.
    This is the case exactly if at least two opposing boundary sides of~$Q$ contain boundary vertices of $P$.
    \begin{figure}[htb]
        \centering%
        \includegraphics[page=4]{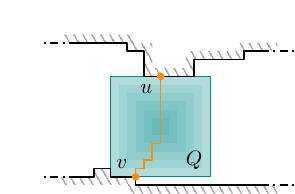}%
        \caption{A maximal square $Q$ with a corresponding geodesic cut.}
        \label{fig:geodesic-cut}
    \end{figure}

    Let~$Q$ be a maximal square and let $u,v$ be vertices on opposing boundary sides of $Q$ with minimal distance to its center, see~\cref{fig:geodesic-cut}.
    Then, there exists a $uv$-path through $Q$ that has length at most $2\depth(P)$.
    We conclude that any maximal square contains a geodesic, non-trivial cut of $P$ that has length at most $2\depth(P)$.
\end{proof}

%% file: 04-upper-bounds.tex
\section{Bounded makespan for narrow instances}
\label{sec:upper-bounds}

In this section, we consider algorithms for bounded makespan in specific families of instances.
Our central result is an approach for asymptotically worst-case optimal stretch in \emph{narrow} instances, which we define as follows.
An instance of \diam~$\diam$ in a polyomino $P$ is \emph{narrow}, if and only if $\pi \cdot \diam \geq \depth(P)$ for some constant $\pi\in \mathbb{N}$, i.e., $\depth(P) \in \BigO(\diam)$.
Intuitively, these correspond to instances of large \diam relative to the domain's \depth.

\begin{restatable}{theorem}{theoremBottleneckStretch}
    \label{thm:bottleneck-stretch}
    Given an instance of \diam $\diam$ in a \solvable polyomino~$P$, we can efficiently compute an applicable schedule of makespan $\BigO(\nicefrac{(\diam+\depth(P))^2}{\bottleneck(P)})$.
    This is asymptotically worst-case optimal for narrow instances.
\end{restatable}

As our proof is fairly involved, we proceed with the special case of \emph{scaled polyominoes} in~\cref{subsec:bounds-for-scaled-polyominoes}, which we extend to arbitrary polyominoes of limited depth in the subsequent~\cref{subsec:bounded-stretch-based-on-bottleneck}.
In each section, we first establish bounds on the makespan relative to a polyomino's area and the corresponding shape parameter.

\input{04a-bounds-scaled}
\input{04b-bounds-bottleneck}

%% file: 04a-bounds-scaled.tex
\subsection{Bounded makespan and stretch based on scale}
\label{subsec:bounds-for-scaled-polyominoes}
We now investigate \emph{scaled polyominoes}, which we define as follows.

\begin{figure}[htb]
    \begin{subfigure}[t]{0.5\textwidth}
        \centering%
        \includegraphics[page=1]{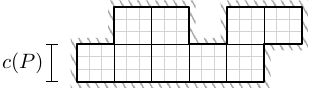}%
        \caption{A $3$-\scale{}d polyomino $P$ and its tiles.}
        \label{fig:scaled-polyomino}
    \end{subfigure}%
    \begin{subfigure}[t]{0.5\textwidth}
        \centering%
        \includegraphics[page=2]{./figures/c-scaled-1.5}%
        \caption{The tile dual graph of $P$.}
        \label{fig:polyomino-tile-dual}
    \end{subfigure}
    \caption{An illustration of a scaled polyomino $P$, its tiles, and their corresponding dual graph.}
\end{figure}

\descriptionlabel{Scaled polyomino.}
For any $\scale\in\mathbb{N}$, we say that a polyomino $P$ is $\scale$-\scale{}d exactly if it is composed of $\scale\times\scale$ squares that are aligned with a corresponding $\scale\times\scale$ integer grid.
We call these grid-aligned squares \emph{tiles}, which have a \emph{dual graph} analogous to that of a polyomino.
Finally, the \emph{\scale} of a polyomino $P$ is the largest integer $\scale(P)$ such that $P$ is $\scale$-\scale{}d.
This additionally represents a very natural lower bound on the \bottleneck, $\bottleneck(P)\geq\scale(P)$.
\medskip

\begin{proposition}
    \label{prop:n-by-scale-sorting}
    For any two configurations of a polyomino $P$ with area $n$ and $\scale(P)\geq3$, we can compute an \applicable schedule of makespan $\BigO(\nicefrac{n}{\scale(P)})$ in polynomial time.
\end{proposition}
\begin{proof}
    We model our problem as an instance of \textsc{Permutation Routing}, taking note of two significant results regarding the routing number of specific graph classes.
    Recall that the routing number $rt(G)$ of a specific graph~$G$ refers to the maximum number of necessary routing operations to transform one labeling of~$G$ into another.
    For the complete graph ${K}_n$ with $n$ vertices, it was shown by Alon, Chung, and Graham~\cite{AlonCG94} that $rt({K}_n) = 2$.
    Furthermore, we make use of a result by Banerjee and Richards~\cite{BanerjeeR17} which states that
    for an $h$-connected graph $G$ and any connected $h$-vertex induced subgraph $G_h$ of $G$, the routing number $rt(G)$ is in $\BigO(rt(G_h)\cdot \nicefrac{n}{h})$.
    They also describe an algorithm that determines a routing sequence that matches this bound.

    Given a polyomino $P$ and two configurations ${C_1,C_2\in\configurations(P)}$, our goal is to define a secondary graph over the vertices of the dual graph $\dual(P) = (V,E)$ such that a routing sequence over this graph can be transformed into a schedule $C_1\rightrightarrows C_2$ of makespan $\BigO(\nicefrac{n}{\scale(P)})$.

    We define $G_\scale = (V,E_\scale)$ such that $\{u,v\}\in E_\scale$ exactly if the cells $u$ and $v$ are located in the same $\scale(P)\times\scale(P)$ tile, or two adjacent tiles.
    As a result, the cells of each tile in $P$ form a clique, i.e., their induced subgraph is isomorphic to $K_{\scale(P)^2}$.
    Furthermore, the cliques of cells in any two adjacent tiles are connected by a set of complete bipartite edges, so they also form a clique.
    Hence, $G_\scale$ is $h$-connected for $h\geq\scale(P)^2 - 1$ and contains $\nicefrac{n}{\scale(P)^2}$ cliques of order at least $\scale(P)^2$.
    Due to Banerjee and Richards~\cite{BanerjeeR17}, we conclude that $rt(G_\scale)$ is in $\BigO(rt(K_{\scale^2})\cdot \nicefrac{n}{h})=\BigO(\nicefrac{n}{\scale(P)^2})$ and can therefore compute a sequence of $\BigO(\nicefrac{n}{\scale(P)^2})$ matchings to route between any two labelings of $\dual(P)$, which correspond to configurations of $P$.

    It remains to argue that we can realize the swaps induced by any matching in $G_\scale$ by means of $\BigO(\scale(P))$ transformations.
    All pairwise swaps between cells within the same tile can be realized by applying \rotatesort to all tiles in parallel, taking $\BigO(\scale(P))$ transformations.

    We therefore turn our attention to swaps between adjacent tiles.
    Observe that the dual graph of the tiles of $P$ is a minor of $G_\scale$; contracting the vertices in each of the tile-cliques defined above will give us a corresponding grid graph.
    Swaps between adjacent tiles can therefore be realized in four phases by covering this grid graph by matchings, and applying \rotatesort to the union of matched tile pairs in parallel, again taking $\BigO(\scale(P))$ transformations.
    A cover by four matchings can be determined by first splitting the edges of the dual graph into two sets of horizontal and vertical edges, respectively.
    Each of these edge sets then induces a collection of paths in the tiling's dual graph, and can therefore be covered by two matchings.

    We conclude that constantly many phases of parallel applications of \rotatesort suffice to realize any matching in $G_\scale$.
    As $\BigO(\nicefrac{n}{\scale(P)^2})$ matchings can route between any two configurations of $P$, we conclude that this method yields schedules of makespan $\BigO(\nicefrac{n}{\scale(P)})$.
\end{proof}

We now apply this intermediate result to compute schedules of bounded stretch in narrow instances of scaled polyominoes.
Our approach hinges on the ability to divide the instance into subproblems that can be solved in parallel, which corresponds to cutting the polyomino and performing a sequence of preliminary transformations such that each subpolyomino can then be reconfigured locally, obtaining the target configuration.

\medskip
\descriptionlabel{Domain partitions.}
A \emph{partition} of a polyomino $P$ corresponds to a set of disjoint subpolyominoes that cover $P$.
We observe that not every polyomino permits a partition into disjoint \solvable subpolyominoes.
However, any subpolyomino of a \solvable polyomino $P$ can be made \solvable by including cells of geodesic distance at most $2$ in $P$.
\medskip
Given a partition of a polyomino, the critical agents to create independent instances are those in close vicinity to the cuts.
To characterize proximity to cuts, we define the following.

\medskip
\descriptionlabel{Cut expansions.}
The \emph{$k$-expansion} of a cut through a polyomino $P$ corresponds to all cells in $P$ from which there exists a geodesic path of length at most $k$ that crosses the cut.

\begin{lemma}
    \label{lem:cut-expansion-area}
    The $k$-expansion of a cut of length $m$ contains $\BigO(k^2 + mk)$ cells.
\end{lemma}
\begin{proof}
    Consider a cut of length $1$.
    Clearly,~$\BigO(k^2)$ cells in any polyomino have incident geodesic paths of length~$\leq k$ that cross this cut.
    Furthermore, extending the cut by a constant distance can only ever increase the area by $\BigO(k)$, as illustrated in~\Cref{fig:cut-expansion}.

    \begin{figure}[htb]
        \hfil%
        \includegraphics[page=1]{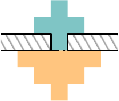}%
        \hfil%
        \includegraphics[page=2]{./figures/cut-expansion}%
        \hfil%
        \includegraphics[page=3]{./figures/cut-expansion}%
        \hfil%
        \caption{The $3$-expansions of cuts of length $m\in\{1,2,3\}$ and the associated difference in area.}
        \label{fig:cut-expansion}
    \end{figure}

    Now consider two adjacent cells $v,w$ in $P$.
    A cell which has geodesic distance at most~$k-1$ from~$v$ is clearly reachable from~$w$ in at most $k$ steps, and vice versa.
    Therefore, the symmetric difference of cells reachable in at most $k$ steps from either cell has a cardinality in $\BigO(k)$.
    This directly implies that the $k$-expansion of a cut of length $m$ can contain at most $\BigO(k)$ cells more than is possible for a cut of length $m-1$, i.e., $\BigO(k^2 + mk)$ cells in total.
\end{proof}

For the purpose of this paper, we compute domain partitions using breadth-first search across the dual graph of the given polyomino, using the following notation.

\medskip
\descriptionlabel{Breadth-first search.}
For any polyomino $Q$, let $\texttt{BFS}(Q,v,r)$ refer to the subpolyomino of~$Q$ that contains all cells reachable from some cell $v$ in $Q$ by geodesic paths of length at most~$r$.
Further, let $\overline{\texttt{BFS}}(Q,v,r)$ refer to the set of connected components of $Q\setminus\texttt{BFS}(Q,v,r)$.
We define the \emph{wavefront} of $\texttt{BFS}(Q,v,r)$ as the set of cuts through $Q$ that define the components of $\overline{\texttt{BFS}}(Q,v,r)$.
Each connected component (cut) of the wavefront is called a \emph{wavelet}.

\begin{lemma}
    \label{lem:bfs-cuts-are-short}
    For any polyomino $Q$, the wavefront of \emph{$\texttt{BFS}(Q,v,r)$} consists of wavelets of length $\BigO(\depth(Q))$ each, i.e., the wavelet length is independent of the search radius~$r$.
\end{lemma}

\begin{proof}
    We argue by induction over the search radius:
    For an arbitrary but fixed cell $v$ and any integer $i\in\mathbb{N}$, each wavelet on the wavefront of $\texttt{BFS}(Q,v,i)$ has length at most $48\depth(Q)$.

    For the base case of $\texttt{BFS}(Q,v,0)$, our induction hypothesis holds true; the wavefront here has length at most $4$.
    Note that the length of a wavefront (which corresponds to the sum of lengths of its wavelets) grows at most linearly in the search radius; it cannot be more than twelve times the search radius.
    Therefore, our hypothesis is also true for any $i\leq 2\depth(Q)$.

    Consider now any $i\geq 2\depth(Q)$ such that the hypothesis is true for all $j\leq i$, i.e., the wavefront of $\texttt{BFS}(Q,v,j)$ does not contain a wavelet that is longer than $48\depth(Q)$ if $j\leq i$.

    Recall~\cref{lem:distance-to-non-trivial-cut} and~\cref{fig:geodesic-cut}:
    Every cell in $Q$ has geodesic distance at most $2\depth(Q)$ to a non-trivial geodesic cut of length at most $2\depth(Q)$.
    This means that there exists some $j\in[i-2\depth(Q),i]$ such that a wavelet~$W$ of $\texttt{BFS}(Q,v,j)$ touches, but does not intersect, a non-trivial geodesic cut $\Gamma$ of length at most $2\depth(Q)$.
    This corresponds to the paths of $W$ and~$\Gamma$ sharing edges, but no two edges of $W$ being separated by $\Gamma$.
    We show that by $j+2\depth(Q)+1$,
    \smallskip
    \begin{description}
        \item[(a)] the corresponding wavelet has completely passed over the cut, resulting in a wavelet~$W'$ that is separated from~$W$ by $\Gamma$, and \item[(b)] $W'$ has length at most~$48\depth(Q)$.
    \end{description}

    \begin{figure}[htb]
        \begin{subfigure}[t]{0.47\textwidth}%
            \centering%
            \includegraphics[page=5]{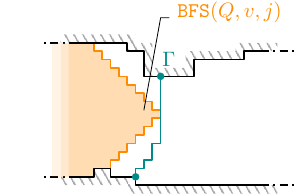}%
            \caption{A wavelet arrives at a geodesic cut...}
            \label{fig:wavelets-are-short-a}
        \end{subfigure}%
        \hfill%
        \begin{subfigure}[t]{0.47\textwidth}%
            \centering%
            \includegraphics[page=6]{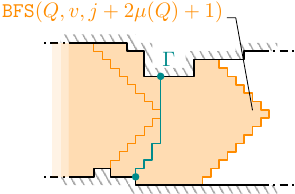}%
            \caption{... and traverses it in $2\depth(Q)$ \texttt{BFS} steps.}
            \label{fig:wavelets-are-short-b}
        \end{subfigure}%
        \caption{An illustration of a wavelet's interaction with short geodesic cuts.}
        \label{fig:wavelets-are-short}
    \end{figure}

    Due to triangle inequality, (a) is trivial:
    $\texttt{BFS}(Q,v,j+1)$ contains at least one cell adjacent to $\Gamma$ on either side.
    As $\Gamma$ has length at most $2\depth(Q)$, the nature of $\texttt{BFS}$ implies that $\texttt{BFS}(Q,v,j+2\depth(Q)+1)$ contains all cells adjacent to $\Gamma$, so $\Gamma$ does not intersect $W'$.

    It remains to argue (b) by bounding the length of the wavelet $W'$ that traversed $\Gamma$.
    As~only $2\depth(Q)$ steps were made from $W$, and $W'$ does not intersect $\Gamma$, $W'$ lies completely in the $2\depth(Q)$-expansion of $\Gamma$.
    Assume that $\Gamma$ separates $Q$ into two polyominoes $Q_1$ and $Q_2$.
    It~now suffices to limit the number of cells in $Q_2$ that have geodesic distance exactly $2\depth(Q)$ to the original wavefront $W$ in $Q_1$.

    We argue on a cell-by-cell basis in $Q_2$.
    Consider a cell $q$ in $Q_2$ that has geodesic distance~$m$ to the cut $\Gamma$, such that
    $q$ has geodesic distance exactly $2\depth(Q)$ to a point on the original wavefront.
    Any geodesic path from $q$ to the closest point $v$ on $W$ crosses $\Gamma$.
    Let $x$ refer to the crossing vertex, as illustrated in~\cref{fig:wavefront-crossing-distance}.
    \begin{figure}[hb]
        \centering
        \includegraphics{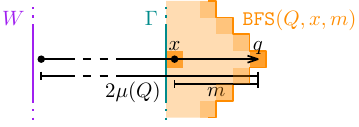}
        \caption{Illustration for~\cref{lem:bfs-cuts-are-short}: A point $q$, separated from $W$ by the cut $\Gamma$.}
        \label{fig:wavefront-crossing-distance}
    \end{figure}

    It follows that all vertices within geodesic distance at most $m-1$ to~$x$ have geodesic distance less than $2\depth(Q)$ to $v$.
    Due to~\cref{lem:cut-expansion-area}, there are only $\BigO(\depth(Q)^2)$ many candidate positions for $q$, of which a number quadratic in $m$ are therefore guaranteed to be closer than~$2\depth(Q)$ to the original wavefront.
    It follows that the number of cells in $Q_2$ with geodesic distance exactly $2\depth(Q)$ to the original wavefront is bounded by~$\BigO(\depth(Q))$, i.e., the wavelet~$W'$ has length~$\BigO(\depth(Q))$.
\end{proof}

Having established all necessary tools, we prove the following statement.
\begin{proposition}
    \label{prop:scale-stretch}
    Given an instance of \diam $\diam$ in a polyomino $P$ with $\scale(P)\geq 3$, we can efficiently compute an applicable schedule with makespan $\BigO(\nicefrac{(\diam+\depth(P))^2}{\scale(P)})$.
    This is asymptotically worst-case optimal for narrow instances.
\end{proposition}
\begin{proof}
    We consider an instance of \diam~$\diam$ in a simple polyomino $P$.
    We proceed in three phases, which we briefly outline before giving an in-depth description of each.
    \smallskip
    \begin{enumerate}[(I)]
        \item We partition $P$ into  $\scale(P)$-\scale{}d \emph{patches} of area $\BigO(\diam^2)$, using non-trivial cuts of bounded length such that the partition's dual graph is a rooted tree $T$.
        \item We combine parent/child patches according to $T$ into regions with $\scale(P)$ \scale, allowing us to apply~\cref{prop:n-by-scale-sorting} to reorder them in $\BigO(\nicefrac{(\diam+\depth(P))^2}{\scale(P)})$.
        \item Finally, we exploit these combined regions to place all agents at their destination.
    \end{enumerate}

    \descriptionlabel{Phase (I).}
    A step-by-step illustration of Phase (I) can be found in~\cref{fig:determining-patches}.
    For this phase, we consider the polyomino $P'$ induced by the tile dual graph of $P$, recall~\cref{fig:polyomino-tile-dual}.
    This scales the shape parameters and geodesic distance by $\nicefrac{1}{\scale(P)}$:
    \[\hfill\scale({P'})=\nicefrac{\scale(P)}{\scale(P)}=1,\qquad \depth({P'})\approxeq\nicefrac{\depth(P)}{\scale(P)},\qquad\bottleneck({P'})\approxeq\nicefrac{\bottleneck(P)}{\scale(P)}.\hfill\]

    Let $\delta = \nicefrac{3\diam}{\scale(P)}$.
    We subdivide ${P'}$ using a recursive breadth-first-search approach and argue by induction.
    Given a boundary cell $v_{0}$ in ${P'}$, we determine a patch ${P'}_0\subseteq {P'}$ based on ${\texttt{BFS}({P'}, v_0, {\delta})\subseteq {P'}}$.
    We say that the components of $\overline{\texttt{BFS}}({P'}, v_0, {\delta})$ are either \emph{small} or \emph{large};
    a~component $R$ is small exactly if $R\subset\texttt{BFS}({P'}, v_0, 2{\delta})$, and large otherwise, see~\cref{fig:gamma-expansion-b}.

    We define ${P'}_0$ as the union of the initial \texttt{BFS} and the small components of its complement, meaning that for $\overline{\texttt{BFS}}({P'}, v_0, {\delta})$ with large components $R_1,\ldots, R_\ell$, ${P'}_0$ takes the shape
    \[\hfill{P'}_0\coloneqq\quad {P'}\setminus (R_1\cup\ldots\cup R_\ell).\hfill\]
    Due to~\cref{lem:bfs-cuts-are-short}, the cut ${\Gamma}_i$ that separates a component $R_i$ of $\overline{\texttt{BFS}}({P'}, v_0, {\delta})$ from ${P'}_0$ has length~$\BigO(\depth({P'}))$.
    By definition, the geodesic distance from each cell in $R_i$ to $v_0$ is at least~${\delta}$.

    We now iteratively subdivide each component of ${P'}\setminus{P'}_0$ by simply increasing the maximal depth of our \texttt{BFS} from $v_0$ by another ${\delta}$ units and again considering large and small components of the corresponding subdivision separately, as illustrated in~\cref{fig:gamma-expansion-c,fig:gamma-expansion-d}.

    To obtain a partition of $P$, we map each patch ${P'}_i$ to the \tiles in $P$ that its cells correspond to.
    Since ${P}$ is a simple polyomino, the dual graph of our patches forms a tree $T$ rooted at~${P}_0$.

    Consider any patch ${P'}_i$ and recall that, due to~\cref{lem:bfs-cuts-are-short}, all cuts induced by \texttt{BFS} have individual length $\BigO(\depth({P'}))$.
    Tracing along the boundaries of tiles, we conclude that the corresponding cuts in $P$ have individual length $\BigO(\depth({P'})\scale(P)))=\BigO(\depth(P))$.
    Due to triangle inequality, it follows that any two cells in each patch $P_i$ have geodesic distance $\BigO(\diam+\depth(P))$.
    From this, we conclude that the area of $P_i$ is bounded by $\BigO((\diam+\depth(P))^2)$.
    It directly follows that for any patch $P_j$ with hop distance~${k\in\mathbb{N}^+}$ to $P_i$ in~$T$, the geodesic distance between two cells in $P_i$ and $P_j$ is bounded by $\BigO(k(\diam+\depth(P)))$.
    The union of patches in a subtree $T'$ of $T$ with height $k$ therefore has area $\BigO((k(\diam+\depth(P)))^2)$.

    \begin{figure}[htb]
        \begin{subfigure}[t]{0.47\textwidth}%
            \centering%
            \includegraphics[page=1]{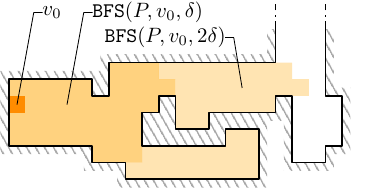}%
            \caption{We determine $v_0$ and compute $\texttt{BFS}({P'}, v_0, \delta)$.}
            \label{fig:gamma-expansion-a}
        \end{subfigure}%
        \hfil%
        \begin{subfigure}[t]{0.47\textwidth}%
            \centering%
            \includegraphics[page=2]{./figures/bfs-cuts}%
            \caption{$\overline{\texttt{BFS}}({P'}, v_0, \delta)$ has large and small components.}
            \label{fig:gamma-expansion-b}
        \end{subfigure}%
        \par\vspace{1em}%
        \begin{subfigure}[t]{0.47\textwidth}%
            \centering%
            \includegraphics[page=3]{./figures/bfs-cuts}%
            \caption{The patch ${P'}_0$ is the union of $\texttt{BFS}({P'}, v_0, \delta)$ and the small components of~$\overline{\texttt{BFS}}({P'}, v_0, \delta)$.}
            \label{fig:gamma-expansion-c}
        \end{subfigure}%
        \hfil%
        \begin{subfigure}[t]{0.47\textwidth}%
            \centering%
            \includegraphics[page=4]{./figures/bfs-cuts}%
            \caption{We continue the breadth-first-search in $P\setminus P_0$.}
            \label{fig:gamma-expansion-d}
        \end{subfigure}%
        \caption{Phase (I): We divide $P'$ into patches of area $\BigO((\delta+\depth(P'))^2)$.}
        \label{fig:determining-patches}
    \end{figure}

    \descriptionlabel{Phase (II).}
    We use this partition of $P$ into patches to subdivide the instance into disjoint tasks that can be solved in parallel; recall that our target makespan is $\BigO(\nicefrac{(\diam+\depth(P))^2}{\bottleneck(P)})$.
    The patches are spatially disjoint and all have \scale at least $\scale(P)$, as well as area $\BigO((\diam+\depth(P))^2)$.
    \cref{prop:n-by-scale-sorting} therefore implies that the patches can be locally reconfigured in parallel, by schedules of makespan $\BigO(\nicefrac{(\diam+\depth(P))^2}{\bottleneck(P)})$.
    In order to solve the original instance, it therefore remains to make each patch a subproblem that can be solved independently.

    We argue that we can efficiently move robots into their target patches.
    In~Phase (I), we gave an upper bound of $\BigO(k(\diam+\depth(P)))$ on the geodesic distance between cells in patches that have hop distance at most $k\in\mathbb{N}^+$ in $T$.
    We now provide a lower bound:
    The~geodesic distance between cells in patches that are not in a parent-child or sibling relationship in $T$ is at least $\diam$, as the distance between cells in any patch and its ``grandparent'' patch according to~$T$ is at least $\diam$ by construction, see Phase (I).

    It follows that the target cell of each robot is either in the same patch as its initial cell, or in a parent or sibling thereof.
    To realize the movement of agents between patches, we thus simply form the spatial union $F_i$ of each patch $P_i$ and its children according to $T$.
    Each of the resulting subpolyominoes $F_{i}$ has area $\BigO((\diam+\depth(P))^2)$.
    As $T$ is bipartite, we can split them into two sets $\mathcal{F}_{A}$ and $\mathcal{F}_{B}$, each comprised of pairwise spatially disjoint subpolyominoes.

    \descriptionlabel{Phase (III).}
    It remains to show that we can efficiently exchange agents between patches.
    Note that, as~$P$ is simple, the number of agents that need to cross any cut in either direction is equal to that for the opposite direction.

    By construction, every pair of patches that needs to exchange agents between one another is fully contained in some $F_i\in(\mathcal{F}_{A}\cup\mathcal{F}_{B})$.
    We proceed in three iterations:
    By applying~\cref{prop:n-by-scale-sorting} to each of the patches in $\mathcal{F}_{A}$ in parallel, we swap agents across cuts by swapping them with agents moving in the opposite direction.
    We repeat this process for $\mathcal{F}_{B}$ and finally perform a parallel reconfiguration of the individual patches, which allows us to place every robot in its target cell.
    Each iteration takes $\BigO(\nicefrac{(\diam+\depth(P))^2}{\bottleneck(P)})$ transformations.
\end{proof}

%% file: 04b-bounds-bottleneck.tex
\subsection{Bounded makespan and stretch based on bottleneck}
\label{subsec:bounded-stretch-based-on-bottleneck}

Finally, this section concerns itself with the transfer of results from~\cref{subsec:bounds-for-scaled-polyominoes} to arbitrary polyominoes.
This requires highly intricate local mechanisms; we start by establishing additional tools and observations.

\medskip
\descriptionlabel{\Skeleton.} A \emph{\skeleton} of a polyomino $P$ is a connected, $\skeleton$-scaled subpolyomino $S\subseteq P$ with $\skeleton = \floor{\nicefrac{\bottleneck(P)}{4}}$, as illustrated in~\cref{fig:skeleton-example}.
Such a \skeleton{} can easily be determined as the union of all $\skeleton\times\skeleton$ squares in $P$ that are aligned with the same $\skeleton\times\skeleton$ integer grid.

\medskip
\descriptionlabel{\Hinterland.}
The \emph{\hinterland} of a \skeleton \tile $t$ corresponds to the union of all $2\skeleton\times 2\skeleton$ squares in $P$ that fully contain $t$, see~\cref{fig:hinterland-example}.
We will show that at least one such square always exists, and their union forms a convex polyomino with \bottleneck at least $\lambda$.

\begin{figure}[htb]
    \begin{subfigure}[t]{0.5\textwidth}%
        \centering%
        \includegraphics[scale=1,page=6]{skeleton-overview}%
        \caption{A polyomino $P$ and its \skeleton $S$ (cyan).}
        \label{fig:skeleton-example}
    \end{subfigure}%
    \hfil%
    \begin{subfigure}[t]{0.5\textwidth}%
        \centering%
        \includegraphics[page=1]{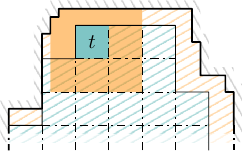}%
        \caption{An illustration of a \skeleton \tile $t$'s \hinterland.}
        \label{fig:hinterland-example}
    \end{subfigure}%
    \hfil%
    \centering%
    \caption{The central tools used in our proof of~\cref{thm:sorting-in-n-by-bottleneck}.}
    \label{fig:bottleneck-sorting-tools}
    \hfil%
\end{figure}
To swap agents from $P\setminus S$ into the \skeleton~$S$, we exploit the \hinterland{}s of its \tiles.
We first prove the existence of a \skeleton $S\subseteq P$ as above, and that its \hinterland{}s fully cover~$P$.
This allows us to apply techniques from~\cref{subsec:bounds-for-scaled-polyominoes} to arbitrary polyominoes.

It is critical that $S$ forms a connected subpolyomino of $P$.
The following can be observed by sliding a square of size  $\bottleneckfrac{2}\times\bottleneckfrac{2}$ along the boundary of $P$:

\begin{observation}
    \label{obs:bottleneck-square-cover}
    We observe that a polyomino $P$ has a maximal cover $\bottleneckHalfCover$ by squares of size $\bottleneckfrac{2}\times\bottleneckfrac{2}$, the center cells of which induce a connected subpolyomino.
\end{observation}

We use this to prove that a connected \skeleton exists and can be determined efficiently.

\begin{lemma}
    \label{lem:skeleton-exists}
    For any polyomino $P$ with $\bottleneck(P)\geq8$, we can identify a \skeleton $S \subseteq P$ in polynomial time.
    The \skeleton $S$ is $\skeleton$-scaled for $\skeleton = \floor{\bottleneckfrac{4}}$, and every $2\skeleton\times 2\skeleton$ square inside~$P$ contains at least one of its scaled tiles.
\end{lemma}
\begin{proof}
    We identify the subpolyomino $S$ by cutting $P$ into \tiles along a grid of size~$\skeleton\times\skeleton$, such that $S$ is the union of all such \tiles that are fully contained in $P$.
    Let $\bottleneckHalfCover$ refer to the maximal cover of $P$ by $2\skeleton\times 2\skeleton$ squares.

    We first show that any square of size $\skeleton \times \skeleton$ inside $P$ must be fully contained in an element of~$\bottleneckHalfCover$.
    Inside any square of size $2\skeleton\times 2\skeleton$, there exist exactly $(\skeleton+1)^2$ unique $\skeleton \times \skeleton$ squares, which form a maximal cover of the square.
    By replacing each square of $\bottleneckHalfCover$ with the corresponding cover, we therefore obtain a cover of $P$ by $\skeleton \times \skeleton$ squares.
    As $P$ is a simple polyomino, this corresponds to all such squares inside~$P$, which implies that every~$\skeleton\times\skeleton$ square in $S$ is contained in a $\bottleneckfrac {2}\times\bottleneckfrac {2}
    $ square in $P$.

    It remains to show that $S$ is a connected polyomino.
    We show this by incrementally selecting from $\bottleneckHalfCover$, while simultaneously marking $\skeleton\times\skeleton$ \tiles induced by the grid that are contained in the selected squares, and therefore also in $S\subset P$.

    Due to~\cref{obs:bottleneck-square-cover}, the center cells of the squares in $\bottleneckHalfCover$ form a connected polyomino.
    Therefore, after starting our construction from an arbitrary square, iteratively adding squares that are offset by one cell either in horizontal or vertical direction, relative to a square of~$\bottleneckHalfCover$ that is already part of the construction, will eventually yield all of $\bottleneckHalfCover$.

    We argue by induction.
    Let $P_i$ refer to the spatial union of previously selected $2\skeleton$-squares from $\bottleneckHalfCover$, and let $S_i$ refer to the set of grid-aligned $\skeleton$-squares in $P_i$.

    For $i=1$, both $P_1$ and $S_1$ form connected polyominoes, so the assumption holds.

    For $P_i \neq P$, there exists a $2\skeleton$-square $Q$ in $P$ that has not been selected yet, but has a center cell adjacent to one of those that define $P_i$.
    We form $P_{i+1}\coloneqq P_i\cup Q$ and simultaneously expand $S_i$ with all $\skeleton$-squares in $Q$ that are aligned with the $\skeleton$ grid, forming $S_{i+1}$.
    As every \skeleton tile is contained in at least one element of $\bottleneckHalfCover$, $P_i=P$ implies that $S_i=S$.

    It remains to argue that if $S_i$ is connected, the same will be true for $S_{i+1}$:
    If no squares are added in step $i$, this is true.
    Therefore, assume that at least one $\skeleton\times\skeleton$ square $D$ is added to $S_{i+1}$.
    This means that $D$ is not fully contained in $P_i$.
    Let~$Q$ now refer to the $2\skeleton$-square added to $P_i$ to form $P_{i+1}$, such that its center cell is adjacent some $2\skeleton$-square $Q'$ in $P_i$.
    As~$D$ is not contained in $Q'$, we conclude it touches the boundary of $Q$.
    Because the side length of~$Q$ is at least twice as long as the side length of~$D$, a second square~$D'$ next to $D$ fits into $Q$.
    As $D'$ is already contained in $Q'$ and therefore in $S_{i}$, we conclude that $S_{i+1}$ is connected.
\end{proof}

As an important secondary result, we conclude that the \hinterland{}s fully cover the polyomino:
Every $\bottleneckfrac{2}\times\bottleneckfrac{4}$ square inside $P$ contains a \skeleton \tile, which in turn has as its \hinterland the union of all such squares that contain it.
\begin{corollary}
    \label{cor:hinterland-cover}
    A polyomino $P$ is fully covered by the union of its \skeleton \tiles' \hinterland{}s.
\end{corollary}

Finally, we demonstrate a method for efficient reconfiguration of an arbitrary \skeleton \tile's \hinterland in order to swap robots into and out of the \skeleton.
\begin{lemma}
    \label{lem:hinterland-reconfiguration}
    Given two configurations of a \skeleton \tile's \hinterland in a \solvable polyomino, we can efficiently compute an \applicable schedule of makespan~$\BigO(\skeleton)$.
\end{lemma}

To prove this, we make use of the following result obtained by Alpert et al.~\cite{ALPERT2022101862} for the routing number of convex grid graphs, where $w(P)$ and $h(P)$ refer to the width and height of $P$, respectively.
\begin{theorem}[Alpert et al.~\cite{ALPERT2022101862}]
    \label{thm:convex-grid-graph-routing}
    Let $P$ be a connected convex grid piece.
    Then the routing number of $P$ satisfies the bound $rt(P)\leq C(w(P) + h(P))$ for some universal constant $C$.
\end{theorem}%

\begin{proof}[Proof of~\cref{lem:hinterland-reconfiguration}]
    Consider a $\skeleton\times\skeleton$ \skeleton \tile $t\subset P$ of a \solvable polyomino $P$, and let $H$ refer to its \hinterland.
    Recall that we assume $\bottleneck(P)\geq 8$, so $\skeleton=\floor{\nicefrac{\bottleneck(P)}{4}}\geq 2$.
    As~$H$ is the union of all $2\skeleton\times2\skeleton$ squares in $P$ that contain $t$, it is thus \solvable and orthoconvex.
    We apply~\cref{thm:convex-grid-graph-routing}:
    Alpert et al.~\cite{ALPERT2022101862} presented a constructive proof in the form of an algorithm, which we can use to compute a routing sequence of length $\BigO(w(H) + h(H))=\BigO(\skeleton)$ between any two configurations of $H$, based on its dual graph.
    Such a routing sequence corresponds to a series of matchings in the dual graph of $H$ that exchange tokens of adjacent vertices.
    Sequentially realizing these matchings by swapping adjacent agents as outlined in~\cref{lem:matching-realization}, we can arbitrarily reorder $H$ in $\BigO(\skeleton)$ transformations.
\end{proof}

This provides us with all necessary tools to prove the following generalization of~\cref{prop:n-by-scale-sorting}, which, in turn, is a central tool for our proof of~\cref{thm:bottleneck-stretch}.

\begin{theorem}
    \label{thm:sorting-in-n-by-bottleneck}
    For any two configurations of a \solvable polyomino $P$ of area $n$, we can compute an \applicable schedule of makespan $\BigO(\nicefrac{n}{\bottleneck(P)})$ in polynomial time.
\end{theorem}
\begin{proof}
    We again argue based on \textsc{Permutation Routing}.
    Given a polyomino $P$ with dual graph $\dual(P) = (V,E)$ and  ${C_1,C_2\in\configurations(P)}$, we define a graph $G_\bottleneck$, such that a routing sequence over $G_\bottleneck$ can be transformed into a schedule $C_1\rightrightarrows C_2$ of makespan $\BigO(\nicefrac{n}{\bottleneck(P)})$.

    To start, we compute a \skeleton $S$ of $P$.
    Then, we group the cells of $P$ into disjoint sets of size $\Theta(\skeleton^2)$ by first grouping those in $S$ based on which \skeleton \tile contains them, see~\cref{fig:skeleton-tiles}, and then adding each remaining cell to the group of one \tile that has a \hinterland which contains it as shown in~\cref{fig:skeleton-grouping}.

    \begin{figure}[htb]
        \hfil%
        \begin{subfigure}[t]{0.47\textwidth}%
            \centering%
            \includegraphics[page=3]{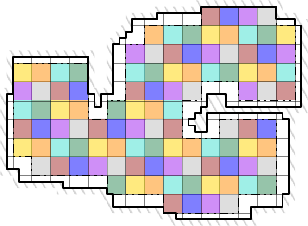}%
            \caption{Grouping cells of $P$ based on \skeleton \tiles.}
            \label{fig:skeleton-tiles}
        \end{subfigure}%
        \hfil%
        \begin{subfigure}[t]{0.47\textwidth}%
            \centering%
            \includegraphics[page=4]{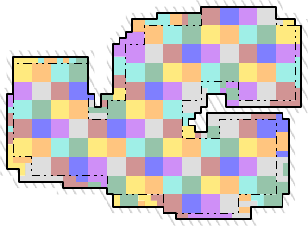}%
            \caption{Grouping cells of $P$ based on \hinterland{}s.}
            \label{fig:skeleton-grouping}
        \end{subfigure}%
        \hfil%
        \caption{We identify the \skeleton of a polyomino, to which we locally assign all remaining cells. }
        \label{fig:skeleton-overview}
    \end{figure}

    We now define $G_\bottleneck = (V,E_\bottleneck)$ such that $\{u,v\}\in E_\bottleneck$ exactly if the cells $u$ and $v$ are located in the same group as defined above, or their groups belong to adjacent \skeleton \tiles.
    As a result, the cells of each group form a clique of order at least $\skeleton^2$.
    Furthermore, the cliques that correspond to adjacent \skeleton \tiles are connected by a set of complete bipartite edges.

    Hence, $G_\bottleneck$ is $h$-connected for $h\geq\skeleton^2 - 1$ and contains $\nicefrac{n}{\skeleton^2}$ cliques of order at least $\skeleton^2$.
    Due to Banerjee and Richards~\cite{BanerjeeR17}, we conclude that $rt(G_\bottleneck)$ is in $\BigO(\nicefrac{n}{\bottleneck(P)^2})$ and can therefore compute a sequence of $\BigO(\nicefrac{n}{\bottleneck(P)^2})$ matchings to route between any two labelings of $\dual(P)$.

    It remains to argue that we can realize all swaps induced by any matching in $G_\bottleneck$ by means of $\BigO(\bottleneck(P))$ transformations.
    We do so in two phases.
    Recall that the \hinterland of each tile has a bounding box no greater than $3\skeleton\times 3\skeleton$, so the same applies to our previously defined clique groups.
    We can therefore sort $\Omega(n)$ of these groups in parallel by applying~\cref{lem:hinterland-reconfiguration}; this takes a total of $\BigO(\skeleton)=\BigO(\bottleneck(P))$ transformations.
    This immediately enables us to perform any swap between cells that are located in the same group, and we turn our attention to swaps between adjacent groups instead.

    Observe that the dual graph of the \skeleton \tiles of $S$ is a minor of $G_\bottleneck$; contracting the vertices in each of the tile cliques will give us a corresponding grid graph.
    Swaps between adjacent tile cliques can therefore be realized in four phases, by covering this grid graph by matchings and sorting the \hinterland{}s of matched groups in unison, thereby exchanging agents between the two.
    As $\BigO(\nicefrac{n}{\bottleneck(P)^2})$ matchings can route between any two configurations of $P$, we conclude that this method yields schedules of makespan $\BigO(\nicefrac{n}{\bottleneck(P)})$.
\end{proof}

The arguments of \Cref{thm:sorting-in-n-by-bottleneck} can also be applied in reverse:
The union of a connected set of \skeleton \tiles and their \hinterland{}s has a large bottleneck, so the following applies.
\begin{corollary}
    \label{cor:sorting-watershed-union}
    A subpolyomino $P'\subset P$ with area $m$ that consists of the \hinterland{}s of a connected set of \skeleton \tiles in a polyomino $P$ can be reordered in $\BigO(\nicefrac{m}{\bottleneck(P)})$ transformations.
\end{corollary}

\theoremBottleneckStretch*
\begin{proof}
    Assume that $\scale(P)<\bottleneck(P)$.
    We briefly outline the three phases of our approach:
    \smallskip
    \begin{enumerate}[(I)]
        \item {
            We compute a \skeleton $S\subset P$ which we split into patches $S_i$ according to Phase~(I) of~\cref{prop:scale-stretch}.
            To each of these patches, we add cells of its \skeleton \tile's \hinterland{}s.
            This partitions $P$ into patches $P_i$, each with a \skeleton patch~$S_i\subset S$ such that $S_i\subseteq P_i$.
        }
        \item {
            We use the rooted dual tree $T$ of the \skeleton patches $S_i$ and,
            for each patch $S_i$ with children $S_\ell, \ldots, S_{\ell+k}$ according to $T$, we combine the patches $P_i$, $P_\ell, \ldots, P_{\ell+k}$ to a (not necessarily connected) region $F_i$ that can be reordered in $\BigO(\nicefrac{(\diam+\depth(P))^2}{\bottleneck(P)})$.
        }
        \item Finally, we exploit these combined regions to place all agents at their destination.
    \end{enumerate}

    \descriptionlabel{Phase (I).}
    We proceed by computing a \skeleton $S\subset P$ of \scale $\skeleton$, which we partition into $\skeleton$-\scale{}d patches using \texttt{BFS}, analogously to Phase (I) of~\cref{prop:scale-stretch}.
    Since~${S}$ is a simple polyomino, the dual graph of our patches $S_i$ forms a tree $T$ rooted at~${S}_0$.

    \begin{figure}[htb]
        \begin{subfigure}[t]{0.47\textwidth}%
            \centering%
            \includegraphics[page=6]{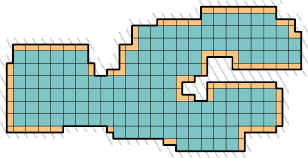}%
            \caption{We compute a skeleton $S\subset P$.}
            \label{fig:partition-based-on-skeleton-a}
        \end{subfigure}%
        \hfill%
        \begin{subfigure}[t]{0.47\textwidth}%
            \centering%
            \includegraphics[page=7]{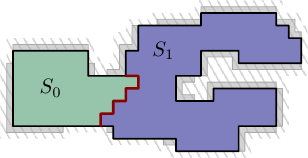}%
            \caption{We partition $S$, see~\cref{prop:scale-stretch}.}
            \label{fig:partition-based-on-skeleton-b}
        \end{subfigure}\par%
        \begin{subfigure}[t]{0.47\textwidth}%
            \centering%
            \includegraphics[page=8]{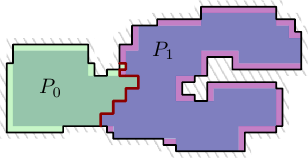}%
            \caption{We add cells of \hinterland{}s to each patch to create a partition of $P$.}
            \label{fig:partition-based-on-skeleton-c}
        \end{subfigure}%
        \hfill%
        \begin{subfigure}[t]{0.47\textwidth}%
            \centering%
            \includegraphics[page=9]{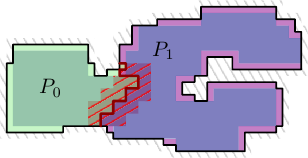}%
            \caption{The \hinterland{}s of these patches may overlap in the highlighted areas at their shared boundary.}
            \label{fig:partition-based-on-skeleton-d}
        \end{subfigure}%
        \caption{We identify the \skeleton of a polyomino, to which we locally assign all remaining cells. }
        \label{fig:partition-approach}
    \end{figure}

    We extend this partition of $S$ to a partition of $P$ by the following process.
    For each patch~$S_i$, we create a patch $P_i$ such that $S_i\subseteq P_i \nsubseteq S$ by greedily assigning cells of~$P\setminus S$ that are located in the watershed of a \skeleton \tile in $S_i$ to $P_i$, as depicted in~\cref{fig:partition-approach}.
    Such a watershed exists for each cell due to~\cref{cor:hinterland-cover}, so the patches $P_i$ split $P$ into disjoint, albeit not necessarily connected, regions.
    For reference's sake, we consider the dual tree $T$ of the patches $S_i$ even in the context of the patches of $P$.
    Recall that there exists a bijection between patches of $S$ and $P$, as $S_i\subseteq P_j$ exactly if $i=j$, and $S_i\cap S_j = \emptyset$ otherwise.

    \descriptionlabel{Phase (II).}
    In order to use this partition of $P$ to achieve bounded stretch, we require the ability to reconfigure the patches in parallel.
    Consider any patch $P_i$ and the corresponding \skeleton patch $S_i$.
    If the entire \hinterland of each \skeleton \tile in $S_i$ were included in $P_i$, we could reconfigure it in $\BigO(\nicefrac{(\diam+\depth(P))^2}{\bottleneck(P)})$ transformations due to~\cref{cor:sorting-watershed-union}.
    However, the \hinterland{}s of $S_i$ can overlap with those of its parent, children, or sibling patches in $T$.
    This makes reconfiguration more involved than in~\cref{prop:scale-stretch}.

    We can, however, easily identify the relevant ``conflict'' zones as cells in \hinterland{}s of \skeleton \tiles in hop distance at most two to a \skeleton \tile in a different patch, see~\cref{fig:partition-based-on-skeleton-d}.
    To deal with these, we define a region $F_i$ for each patch $P_i$, based on its corresponding \skeleton patch $S_i$ and the dual tree $T$.

    Let $F^S_i\subseteq S$ refer to the union of~$S_i$ with its child patches in $S$ according to $T$.
    To~obtain~$F_i$, we include the \hinterland{}s of all \skeleton \tiles in $F^S_i$ that have hop distance at least $2$ to a \skeleton \tile of $S\setminus F^S_i$ in the \tile dual graph of $S$.
    This process is illustrated in~\cref{fig:combining-patches}.

    \begin{figure}[htb]
        \begin{subfigure}[t]{0.3\textwidth}%
            \centering%
            \includegraphics[page=4]{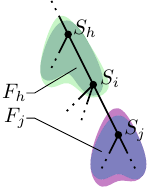}%
            \caption{Part of the dual tree $T$ of \skeleton patches in $S\subseteq P$.}
            \label{fig:combining-patches-tree}
        \end{subfigure}%
        \hfill%
        \begin{subfigure}[t]{0.3\textwidth}%
            \centering%
            \includegraphics[page=2]{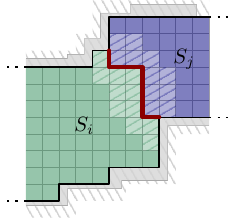}%
            \caption{We exclude \tiles that could induce conflicts, dashed here.}
            \label{fig:combining-patches-a}
        \end{subfigure}%
        \hfill%
        \begin{subfigure}[t]{0.3\textwidth}%
            \centering%
            \includegraphics[page=3]{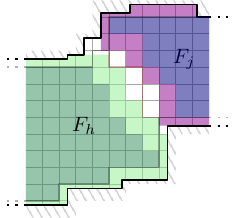}%
            \caption{The remaining \skeleton \tiles' \hinterland{}s form $F_h$ and $F_j$.}
            \label{fig:combining-patches-b}
        \end{subfigure}%
        \caption{The creation of regions $F_i$ and $F_h$: Local conflict avoidance based on the dual tree $T$.}
        \label{fig:combining-patches}
    \end{figure}

    This region $F_i\subset P$ is therefore a connected subpolyomino of $P$ with area $\BigO((\diam+\depth(P))^2)$ and \bottleneck at least $\skeleton = \floor{\bottleneckfrac{4}}$, so it can be reconfigured in $\BigO(\nicefrac{(\diam+\depth(P))^2}{\bottleneck(P)})$.

    For two patches $S_i$ and $S_j$ such that $S_i$ is the ``grandparent'' of $S_j$ according to $T$, the regions $F_i$ and $F_j$ are spatially disjoint.
    As $T$ is bipartite, we can divide the regions $F_i$ into two sets $\mathcal{F}_{A}$ and $\mathcal{F}_{B}$, each comprised of pairwise spatially disjoint subpolyominoes.

    \descriptionlabel{Phase (III).}
    It remains to show that we can efficiently exchange agents between patches.
    By construction, every pair of patches that needs to exchange agents between one another is fully contained in some $F_i\in(\mathcal{F}_{A}\cup\mathcal{F}_{B})$.
    We proceed in three iterations:
    By applying \cref{thm:sorting-in-n-by-bottleneck} to each of the patches in $\mathcal{F}_{A}$ in parallel, we swap agents across cuts by swapping them with agents moving in the opposite direction.
    We repeat this process for $\mathcal{F}_{B}$ and finally perform a parallel reconfiguration of the individual patches, which allows us to place every robot in its target cell.
    Each iteration takes $\BigO(\nicefrac{(\diam+\depth(P))^2}{\bottleneck(P)})$ transformations.
\end{proof}

%% file: 05-conclusion.tex
\section{Conclusions and future work}

We provide a number of novel contributions for \textsc{Multi-Agent Path Finding} in simple polyominoes.
We establish a characterization for the existence of reconfiguration schedules, based on different shape parameters of the bounding polyomino.
Furthermore, we establish algorithmic methods that achieve worst-case optimal stretch for any instance in which the polyomino's \bottleneck does not exceed the instance's diameter by more than a constant factor.
There are a variety of directions in which these insights should give rise to further generalizations and applications.

\medskip
\descriptionlabel{Non-simple polyominoes.}
Our results regarding \solvability are directly applicable to non-simple polyominoes.
As noted in~\cref{sec:reconfigurability}, the geometric characterization for simple polyominoes is formed as a special case based on the dual graph of a polyomino.

For any non-simple polyomino that is either $2$-scaled or $2$-square-connected,~\cref{thm:parallel-bubblesort} and~\cref{prop:n-by-scale-sorting} are also directly applicable.
The same is not true for~\cref{thm:sorting-in-n-by-bottleneck}, as our definition of the \bottleneck based on cuts does not work in this case.
However, with a separate definition that accounts for the minimal distance between inner and outer boundaries, \cref{thm:sorting-in-n-by-bottleneck} may be applicable.

\medskip
\descriptionlabel{Permutation routing.}
Our results can be generalized to solid grid graph routing, which is a generalization of the findings of Alpert et al.~\cite{ALPERT2022101862}.
We provided results regarding bounded stretch for this setting, therefore tackling a special case of their Open Question 2.

\medskip
\descriptionlabel{Further questions.}
Our work is orthogonal to that of Demaine et al.~\cite{dfk+-cmprs-18,dfk+-cmprs-19}:
Their setting considered domains of large \depth in conjunction with large \bottleneck, i.e., the case that $\depth(P)\in\Omega(\diam)$ and $\bottleneck(P)\in\Omega(\diam)$.
We establish asymptotically worst-case optimal results for narrow domains, which implies that $\depth(P)\in\BigO(\diam)$ and $\bottleneck(P)\in\BigO(\diam)$.
In~particular, instances where $\bottleneck(P)\in\smallo(\diam)$ while $\depth(P)\in\omega(\diam)$, i.e., instances in which the gap between \bottleneck and \depth is unbounded relative to $\diam$, remain a challenge even for simple domains.
We conjecture that this question for simple domains is equivalent to that of bounded stretch for non-simple domains with limited \depth;
considering an instance of large \depth, we can create an analogous non-simple instance that features regularly distributed, small holes based on some grid graph.
This may motivate research into the special case of instances in which the \diam is less or equal to the circumference of the smallest hole in the domain.

All these questions remain for future work.